\documentclass[a4paper,11pt]{article}
\pdfoutput=1
\RequirePackage{ifluatex}
\usepackage{fullpage}
\usepackage{array}
\usepackage{xspace}
\usepackage[utf8]{inputenc}
\usepackage[T1]{fontenc}
\usepackage{graphicx}
\usepackage{amsmath,bm}
\usepackage{float}
\usepackage{amsfonts}
\usepackage{algorithm}      
\usepackage[noend]{algpseudocode}  
\usepackage{algorithmicx}
\usepackage[dvipsnames,usenames]{color}
\usepackage[colorlinks=true,urlcolor=Blue,citecolor=Green,linkcolor=BrickRed]{hyperref}
\usepackage[usenames,dvipsnames]{xcolor}
\usepackage{authblk}
\usepackage{amsthm}
\usepackage{bigstrut}
\usepackage[noadjust]{cite}
\usepackage{todonotes}
\usepackage{thmtools,thm-restate}
\usepackage{lineno}
\usepackage{comment}

\newcommand{\Oh}{\mathcal{O}}

\usepackage{tabularx} 
\newcolumntype{C}{>{\centering\arraybackslash}X} 

\newtheorem{theorem}{Theorem}
\newtheorem{lemma}[theorem]{Lemma}
\newtheorem{conjecture}[theorem]{Conjecture}

\newtheorem{corollary}[theorem]{Corollary}

\title{Conditional Lower Bounds for Variants of Dynamic LIS}

\author[ ]{Paweł Gawrychowski\thanks{\texttt{\href{mailto:gawry@cs.uni.wroc.pl}{gawry@cs.uni.wroc.pl}}}}
\author[ ]{Wojciech Janczewski\thanks{\texttt{\href{mailto:wojciech.janczewski@cs.uni.wroc.pl}{wojciech.janczewski@cs.uni.wroc.pl}}}}
\affil[ ]{University of Wrocław}

\date{}

\begin{document}

\maketitle

\begin{abstract}
In this note, we consider the complexity of maintaining the longest increasing subsequence (LIS) 
of an array under (i) inserting an element, and (ii) deleting an element of an array.
We show that no algorithm can support queries and updates in time $\Oh(n^{1/2-\epsilon})$ and
$\Oh(n^{1/3-\epsilon})$ for the dynamic LIS problem, for any constant $\epsilon>0$, when the elements are weighted
or the algorithm supports 1D-queries (on subarrays), respectively, assuming the All-Pairs Shortest Paths (APSP)
conjecture or the Online Boolean Matrix-Vector Multiplication (OMv) conjecture.
The main idea in our construction comes from the work of Abboud and Dahlgaard [FOCS 2016], who proved
conditional lower bounds for dynamic planar graph algorithm.
However, this needs to be appropriately adjusted and translated
to obtain an instance of the dynamic LIS problem.
\end{abstract}


\section{Introduction}

\paragraph{LIS}
Computing the length of a longest increasing subsequence (LIS) is one of the basic algorithmic
problems. Given a sequence $(a_{1},a_{2},\ldots,a_{n})$ with a linear order on the elements, an increasing subsequence
is a sequence of indices $1\leq i_{1} < i_{2} < \ldots < i_{\ell}\leq n$ such that $a_{i_{1}} < a_{i_{2}} < \ldots < a_{i_{\ell}}$.
We seek the largest $\ell$ for which such a sequence exists.
The classic solution~\cite{Fredman} works in time $\Oh(n\log{n})$.

Very recently, starting with the paper of Mitzenmacher and Seddighin~\cite{MitzenmacherS20},
the problem of dynamic LIS was explored.
In a dynamic version of LIS, we need to maintain the length of LIS of an array
under (i) inserting an element, and (ii) deleting an element of an array.
In~\cite{GawrychowskiJ21}, an algorithm working in polylogarithmic time and providing
$(1+\epsilon)$-approximation was presented, and in~\cite{KociumakaS21} Kociumaka and Seddighin
designed a randomised exact algorithm with the update and query time $\tilde{\Oh}(n^{4/5})$.
As for now, no lower bounds are known for the dynamic LIS, besides trivial $\Omega(\log{n})$
in the comparison model.

\paragraph{LIS variants.}
In the variant of LIS with 1D-queries, the algorithm needs to answer queries about the length of LIS
in a subarray $(a_{i},a_{i+1},\ldots,a_{j})$, given any $i,j$.
For static LIS, it is known how to build a structure of size $\Oh(n\log n)$ capable of providing
exact answers to such queries in $\Oh(\log n)$ time using the so-called unit-Monge matrices~\cite[Chapter 8]{Tiskin}.
Interestingly, both existing algorithms for dynamic approximate (global) LIS~\cite{GawrychowskiJ21,KociumakaS21} in fact
support such 1D-queries within the same time complexity,
although that is not the case for the dynamic exact LIS~\cite{KociumakaS21}.

In the weighted variant of LIS, elements have assigned nonnegative weights and we seek an
increasing subsequence with the maximum sum of weights of elements.
Note that in a dynamic weighted LIS changing the weight of an element can be simply done by first deleting and
then inserting a new element.
For the static LIS, the weighted case is also solvable in $\Oh(n\log{n})$ time with a simple algorithm.

Any algorithm for weighted LIS can easily handle 1D-queries.
Let $M$ be the maximum weight of any element of the array of $n$ elements.
Then to answer $(i,j)$ 1D-query, we can insert immediately to the left of $a_i$ an element with weight $M(n+1)$
and being minimum in the linear order of elements, and insert immediately to the right of $a_j$ an element
with weight $M(n+1)$ and being maximum in the linear order of elements.
The global maximum weight of an increasing subsequence in this modified array
is then the weight of an increasing subsequence in the original array on elements between $a_{i}$ 
and $a_{j}$ only, increased by $2M(n+1)$.

\paragraph{Our results.}
In this note, we apply the construction of Abboud and Dahlgaard~\cite{AbboudD16}, originally designed
for distances in planar graphs, to provide conditional polynomial lower bounds for dynamic LIS with 1D-queries
(so also for dynamic weighted LIS). The main idea is to simulate their grid gadgets by an appropriately
designed set of points in the plane, and arrange the points in an array in the natural left-to-right
order. Then, we are able to extract the original distance by querying for the LIS in a subarray.

The recent algorithms for dynamic ($1+\epsilon$)-approximation of
(global) LIS~\cite{MitzenmacherS20,GawrychowskiJ21,KociumakaS21} support 1D-queries without any modifications.
On the other hand, the recent sublinear algorithm for dynamic \emph{exact} LIS~\cite{KociumakaS21} does not.
We hope that our polynomial lower bound for this natural generalisation
helps to understand why there is such a large gap between the polylogarithmic
complexity time for dynamic ($1+\epsilon$)-approximate LIS~\cite{GawrychowskiJ21} and 
$\tilde{\Oh}(n^{4/5})$ time complexity for dynamic exact LIS~\cite{KociumakaS21}.

\section{Preliminaries}
In this work, we consider sets of 2D points.
A point $p_1=(x_1,y_1)$ is dominated by $p_2=(x_2,y_2)$,
denoted $p_1 \prec p_2$, if $x_1 < x_2$ and $y_1 < y_2$.
For any set of points $S$, its ordered subset of points $P=(p_0,p_1,\ldots,p_k)$ is a chain
if $p_0 \prec p_1 \prec \ldots \prec p_k$.
When points have weights, the weight of a chain is the sum of weights of its points.
$P=(p_0,p_1,\ldots,p_k)$ is an antichain if for all $0 \leq i,j \leq k$, $p_i \not \prec p_{j}$.

When we wish to compute the longest chain in a set of points, we simply arrange
its points in an array, sorted by the $x$-coordinates, and then compute LIS with respect to the $y$-coordinates.
Maintaining such an array for a dynamic set of points is easily realised in logarithmic time with any BST
supporting $\mathsf{FindRank}(r)$ operation, which returns a pointer to an item of rank $r$ in the set of items in BST.
Then we can find the value of the current element at index $r$, delete the current element at index $r$,
or insert a new element strictly before or after the current element at index $r$, all with logarithmic overhead.

\paragraph{Two popular conjectures.}
For the weighted case, we assume the following conjecture on the well-known APSP problem:

\begin{conjecture}\label{con:apsp}
There exists no algorithm solving the all-pairs shortest paths (APSP) problem in general weighted graphs
in time $\Oh(n^{3-\epsilon})$, for any constant $\epsilon > 0$.
\end{conjecture}

We reduce from the (max,$+$)-matrix-product problem, denoted by $\circ$,
which is known to be equivalent to APSP~\cite{FischerM71}. 
In this problem we are given $n \times n$ matrices $A$, $B$ having integer weights in $\{0,\ldots,M\}$
and want to compute matrix $C$, with $C_{i,j}=\max_{k}(A_{i,k}+B_{k,j})$.
As those matrices are weighted, for the unweighted case we need a weaker conjecture
on Boolean variables.
In the online Boolean matrix-vector multiplication (OMv) problem we are given $n \times n$ matrix $A$,
and let $v_1,\ldots,v_n$ be Boolean vectors arriving online.
We need to preprocess $A$ and then for every $i$ output $Av_i$ before seeing $v_{i+1}$.
It was conjectured in~\cite{HenzingerKNS15} (see also~\cite{LarsenW17,ChakrabortyKL18}) that:

\begin{conjecture}\label{con:omv}
There exists no algorithm solving the OMv problem in time $\Oh(n^{3-\epsilon})$, for any constant $\epsilon > 0$.
\end{conjecture}

\section{The embedding of a Matrix}
Assume for now that matrices $A$, $B$ are Boolean.
Columns and rows are numbered from $0$ to $n-1$.
For each $0 \leq k < n$, we define embedding $S_k$ of $A$ and the $k$-th column of $B$, $B_k$,
as the union of six sets of points, namely $S_k = S_l \cup S_l' \cup S_r \cup S_a \cup S_a' \cup S_{b,k}$.
In total, $S_k$ contains $3n^2+3n$ points.
It consists of the left and right side, and three sets of special points.
It is best to inspect Figure~\ref{fig:emb} before reading further.

\begin{figure}
\begin{center}
  \includegraphics[width=\linewidth]{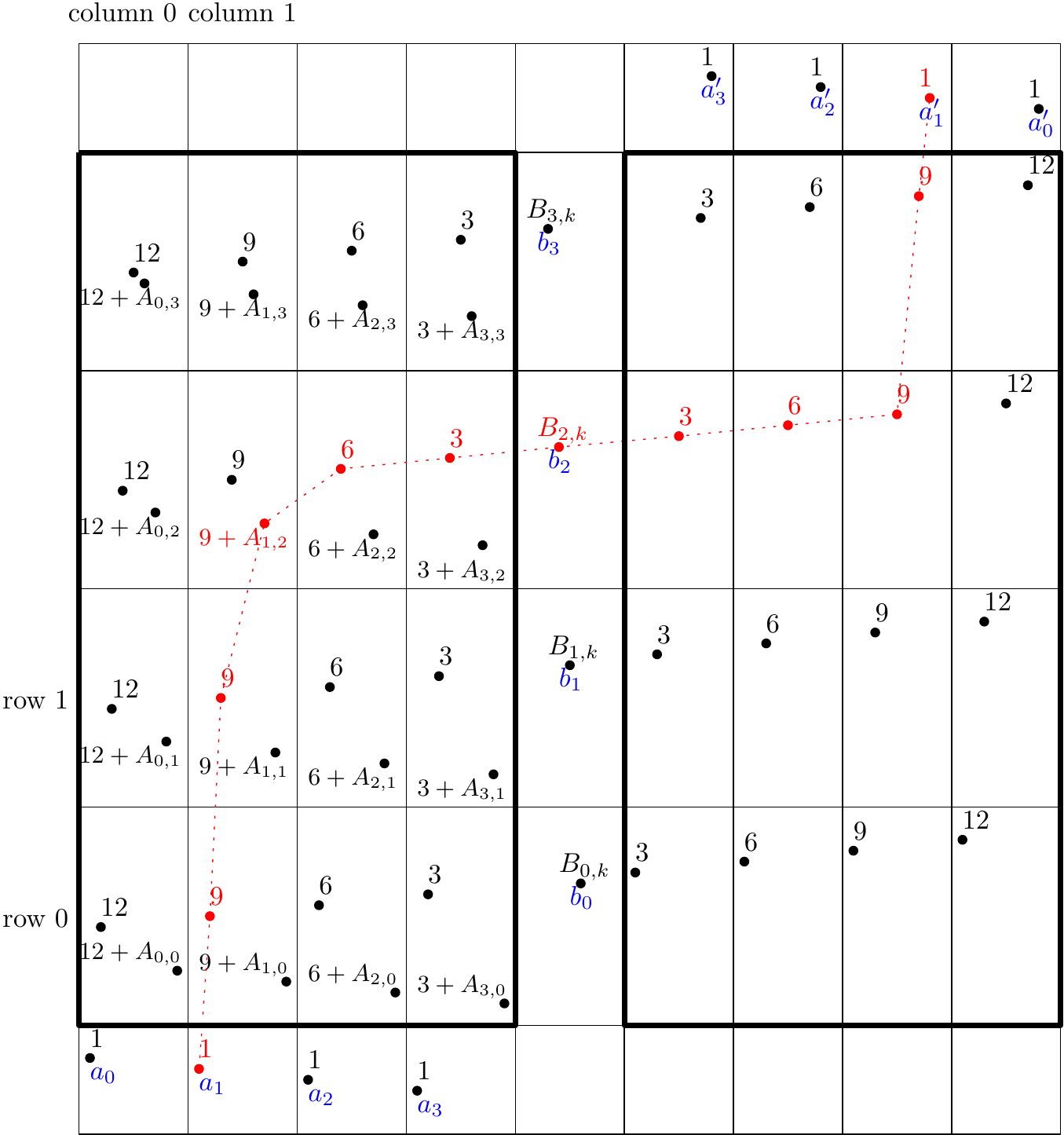}
\end{center}
  \caption{Embedding of matrix A and the $k$-th column of matrix B, with $n=4$.
  Each point is given an integer weight.
  Special points $a_0,\ldots,a_{n-1},b_0,\ldots,b_{n-1},a'_0,\ldots,a'_{n-1}$ are marked in blue.
  On the left side, weights of points use weights from transposed $A$.
  Turns allow to gain $A_{i,j}$, but taking them too early is not beneficial due to
  diminishing weights in columns.
  In red, the maximum weight chain from $a_1$ to $a'_1$ is highlighted.
  The right side of the embedding makes the choice of $b_i$ dependent only on values from matrices.}
  \label{fig:emb}
\end{figure}

On the left side, we imagine there is a $n \times n$ grid.
There are two sets of points $S_l$ and $S_l'$, both having one point in each cell of this grid.
For each column or row, points of $S_l$ form one chain, while points from $S_l'$ form one antichain.
Additionally, in each cell the point from $S_l$ is above and to the left of the point from $S_l'$.
As for the weights, in column $j$ points from $S_l$ have weights $3(n-j)$, and in cell $i,j$ the point
from $S_l'$ has the weight $3(n-j)+A_{j,i}$, so depending on the value from transposed matrix $A$.
Formally, for each $0 \leq i,j < n$, $S_l$ contains the point $l_{i,j}=(j(2n+1)+i+2,i(3n+1)+2n+j+1)$
with the weight $3(n-j)$, and $S_l'$ contains the point
$l'_{i,j}=((j+1)(2n+1)-i,i(3n+1)+2n-j)$ with the weight $3(n-j)+A_{j,i}$.
The first special set of points $S_a$ is also defined using the left grid.
All points from $S_a$ are below points from $S_l$ and $S_l'$, additionally $j$-th point is to the left
of all points from $j$-th column of the grid and to the right of points from $(j-1)$-th column.
All these points have the weight $1$ and form an antichain.
Formally, for each $0 \leq j < n$, $S_a$ contains a point $a_j=(j(2n+1)+1,n-j)$ with the weight $1$.

On the right side, we also have $n \times n$ grid and the set of points $S_r$, which is basically
$S_l$ rotated by $180$ degrees.
In each cell of the grid, there is one point from $S_r$, and points in columns or rows form chains,
but weights in column $j$ are $3(j+1)$.
Formally, for each $0 \leq i,j < n$, $S_r$ contains a point $((2n+j)(n+1)+i+1,(i+1)(3n+1)+j+1)$ with 
the weight $3(j+1)$.
The second special set of points $S_a'$ is defined using the right grid.
All points from $S_a'$ are above points from $S_r$ and additionally, $j$-th point is to the right
of all points from $j$-th column of the grid and to the left of points from $(j+1)$-th column.
All these points have the weight of $1$ and form an antichain.
Formally, for each $0 \leq j < n$, $S_a'$ contains a point $a'_{n-j-1}=((2n+j+1)(n+1),3n(n+1)-j)$ with 
the weight $1$.

The last special set of points is $S_{b,k}$.
Those points form one antichain and are placed between the left and right grids.
$i$-th point is above all the points from the $i$-th row from the left grid
and below all the points from the $i$-th row of the right grid.
The weight of $i$-th point is just $B_{i,k}$.
Formally, for each $0 \leq i < n$, $S_{b,k}$ contains a point $b_i=(2n(n+1)-i,(i+1)(3n+1))$
with the weight $B_{i,k}$.

Note that for $k \neq k'$, embeddings $S_k$ and $S_k'$ differ only on weights of $n$ points from $S_{b,k}$ and $S_{b,k'}$.

\paragraph{Properties of the longest chains.}
Given $S_k$ we need to inspect the maximum weight of a chain starting with $a_j$ and ending with $a'_j$,
for any $j$.
Observe that there is such a maximum chain containing a point from $S_{b,k}$,
thus we can focus on chains from $a_j$ to $b_i$ and from $b_i$ to $a'_j$, for any $i$.

The main idea here is that the maximum weight chain should use at most one point from $S'_l$.
Those points can be seen as 'turns', meaning that if point $p$ in column $j$ from $S'_l$ is taken,
then the chain cannot contain any point from column $j$ above $p$.
It is not beneficial to take such a turn too quickly, since by changing a column early one can gain at most 2
(from two non-zero entries in a matrix), at the same time losing at least 3 from diminishing weights in columns.

Let $c_k(p_1,p_2)$ be the maximum weight of a chain starting with a point $p_1$ and ending with
a point $p_2$, in set $S_k$.
Then we have:

\begin{lemma}\label{lem:cases}
Consider any $0 \leq i,i',j,k < n$, with $i \leq i'$.
\begin{enumerate}
\item If $i = i'$, $c_k(l_{i,j},b_{i'})=1.5(n-j)(n-j+1)+B_{i',k}$.
\item If $i < i'$, $c_k(l_{i,j},b_{i'})=(3n-3j)(i'-i)+1.5(n-j)(n-j+1)+A_{j,i'}+B_{i',k}$.
\item If $i = i'$, $c_k(l'_{i,j},b_{i'})=1.5(n-j)(n-j+1)+A_{j,i}+B_{i',k}$.
\item If $i < i'$, $c_k(l'_{i,j},b_{i'})=(3n-3j-3)(i'-i)+1.5(n-j)(n-j+1)+A_{j,i}+A_{j+1,i'}+B_{i',k}$, where we assume $A_{n,\cdot}=0$.
\end{enumerate}
\end{lemma}

\begin{figure}[h]
\begin{center}
  \includegraphics[scale=0.8]{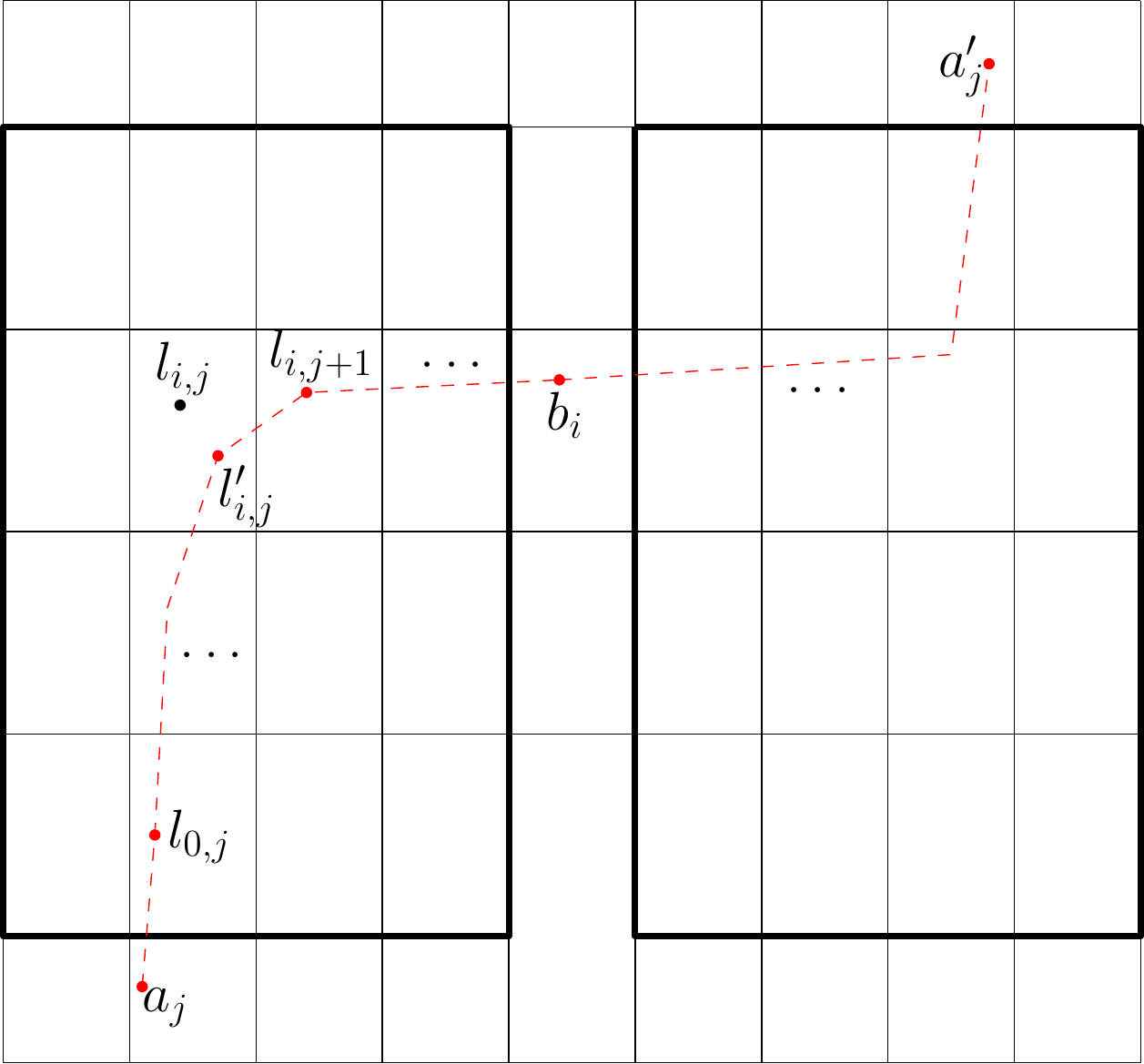}
\end{center}
  \caption{Simplified representation of the maximum chain from $a_j$ to $a'_j$.}
  \label{fig:mini}
\end{figure}

\begin{proof}
The proof, for any $k$ and $i'$, is by induction on decreasing $j$ and then decreasing $i$.
For $j=n-1$ the claim is easy to verify --- the maximum chain from $l'_{i,n-1}$ to $b_{i'}$ contains only two points,
and from $l_{i,n-1}$ to $b_{i'}$ it is a chain containing, except endpoints,
$i'-i-1$ points from $S_l$ in a column $n-1$, and $l'_{i',n-1}$.
The case of $i=i'$ is also straightforward to verify, as then we have just a chain in a single row of the grid.

Now assume $i<i',j<n-1$.
For a maximum chain starting with $l'_{i,j}$, the next point in the chain must be $l_{i,j+1}$,
thus by induction the maximum weight is
\begin{align*}
&(3(n-j)+A_{j,i})+(3n-3j-3)(i'-i)+1.5(n-j-1)(n-j)+A_{j+1,i'}+B_{i',k} = \\
&(3n-3j-3)(i'-i)+1.5(n-j)(n-j+1)+A_{j,i}+A_{j+1,i'}+B_{i',k},
\end{align*}
which gives us equality 4.

Now assume $i=i'-1,j<n-1$.
For a maximum chain starting at $l_{i,j}$, the next point in the chain must be $l_{i',j}$,
$l'_{i',j}$ or $l_{i,j+1}$.
Thus, by induction its weight is
\begin{align*}
& 3(n-j)+ \max\{1.5(n-j)(n-j+1)+B_{i',k},\\
&1.5(n-j)(n-j+1)+A_{j,i'}+B_{i',k},\\
&(3n-3j-3)+1.5(n-j-1)(n-j)+A_{j+1,i'}+B_{i',k}\} = \\
& 3(n-j)+1.5(n-j)(n-j+1)+B_{i',k} +\max\{0,A_{j,i'},A_{j+1,i'}-3\}=\\
&(3n-3j)(i'-i)+1.5(n-j)(n-j+1)+A_{j,i'}+B_{i',k}.
\end{align*}

Now assume $i<i'-1,j<n-1$.
For a maximum chain starting at $l_{i,j}$, the next point in the chain must be $l_{i+1,j}$,
$l'_{i+1,j}$ or $l_{i,j+1}$.
Thus, by induction its weight is
\begin{align*}
& 3(n-j)+ \max\{(3n-3j)(i'-i-1)+1.5(n-j)(n-j+1)+A_{j,i'}+B_{i',k}, \\
& (3n-3j-3)(i'-i-1)+1.5(n-j)(n-j+1)+A_{j,i+1}+A_{j+1,i'}+B_{i',k}, \\
& (3n-3j-3)(i'-i)+1.5(n-j-1)(n-j)+A_{j+1,i'}+B_{i',k}\} = \\
& 3(n-j)+(3n-3j)(i'-i-1)+1.5(n-j)(n-j+1)+B_{i',k}
\\&+\max\{A_{j,i'},3i-3i'+3+A_{j,i+1}+A_{j+1,i'},3i-3i'+A_{j+1,i'}\}=\\
&(3n-3j)(i'-i)+1.5(n-j)(n-j+1)+A_{j,i'}+B_{i',k},
\end{align*}
as $i'-i>1$ and matrices are Boolean.
Thus, we get the remaining equality 2.
\end{proof}

From Lemma~\ref{lem:cases} we immediately have:
\begin{corollary}\label{lem:paths1}
For any $0 \leq i,j,k < n$, $c_k(a_j,b_i)=(3n-3j)i+1.5(n-j)(n-j+1)+1+A_{j,i}+B_{i,k}$.
\end{corollary}

The next lemma for weights in the simpler right grid is straightforward to prove:
\begin{lemma}\label{lem:paths2}
For any $0 \leq i,j,k < n$, $c_k(b_i,a'_j)=(3n-3j)(n-i-1)+1.5(n-j)(n-j+1)+1+B_{i,k}$.
\end{lemma}

Observe that in $S_k$ there is a maximum chain from $a_j$ to $a'_j$ containing a point from $S_{b,k}$,
as there is a point from $S_{b,k}$ between any two points from the left and right grid forming a chain.
Thus, from Corollary~\ref{lem:paths1} and Lemma~\ref{lem:paths2}, we have:
\begin{lemma}
For any $0 \leq j,k < n$, $c_k(a_j,a'_j)=\max_{i}\{(3n-3j)(n-1)+3(n-j)(n-j+1)+2+A_{j,i}
+B_{i,k}\}=(3n-3j)(n-1)+3(n-j)(n-j+1)+2+\max_{i}(A_{j,i}+B_{i,k})$.
\end{lemma}

Therefore, we achieved that for the maximum chain from $a_j$ to $a'_j$ going through $b_i$,
the only dependence on $i$ are values from matrices.
Moreover, observe that in any $S_k$ the maximum weight chain on points with $X$-coordinates between
$a_j$ and $a'_j$ is $c_k(a_j,a'_j)$, that is, such a chain starts with $a_j$ and ends with $a'_j$.
This is true because any chain can contain at most one point from $S_a$ and at most one point from $S'_a$,
and when we consider chains on points with $X$-coordinates between $a_j$ and $a'_j$, we can always either
replace some point $a_{j'}$ with $j'>j$ by $a_j$ or just add $a_j$ as the first element of a chain,
and similarly for $a'_j$.
This will allow us to compute (max,$+$)-product of matrices using LIS supporting 1D-queries.

\paragraph{Weighted matrices.}
The previous embedding for Boolean matrices extends naturally to the case of
$A$ and $B$ having integer weights in $\{0,\ldots,M\}$ for some $M \in \mathrm{poly}(n)$.
We only need to multiply the previous weights of points in sets $S_l$ and $S_r$ by $M$,
and change weights of point $l'_{i,j}$ in set $S_l'$ to $3M(n-j)+A_{j,i}$.
All logic on the maximum chains still applies, and in the weighted case the following can be proven:

\begin{lemma}\label{lem:weighted}
For any $0 \leq j,k < n$, $c_k(a_j,a'_j)=M(3n-3j)(n-1)+3M(n-j)(n-j+1)+2+\max_{i}(A_{j,i}+B_{i,k})$.
\end{lemma}

\section{Lower Bounds}
\paragraph{Dynamic weighted LIS.}
Using embeddings defined in the previous section,
it is straightforward to obtain a lower bound for weighted LIS:

\begin{theorem}
There is no algorithm for dynamic weighted LIS with the amortised update and query time $\Oh(n^{1/2-\epsilon})$
for any constant $\epsilon>0$, unless Conjecture~\ref{con:apsp} is false.
This holds even if only weight updates are allowed.
\end{theorem}
\begin{proof}
We reduce from (max,$+$)-product using defined embeddings.
We are given $n \times n$ matrices $A$, $B$ having integer weights in $\{0,\ldots,M\}$
and want to compute $C$, with $C_{i,j}=\max_{k}(A_{i,k}+B_{k,j})$.
Performing $\Oh(n^2)$ updates, we insert all the points from embedding $S_1$ of $A$ and $B_1$.
Then we make $n$ queries, for each $0 \leq j < n$ asking about the maximum weight of a chain on points
with $X$-coordinates between $a_j$ and $a'_j$.
From Lemma~\ref{lem:weighted}, this is $M(3n-3j)(n-1)+3M(n-j)(n-j+1)+2+\max_{i}(A_{j,i}+B_{i,1})$.
As only the last part of this formula is not fixed, we can extract from the answer the value of $C_{j,1}$.
Thus, after $n$ queries we have the first column of $C$ computed.
Then we update weights of points in $S_{b,1}$ to get $S_{b,2}$, which transforms $S_1$ into $S_2$.
This process is repeated until we get the whole $C$.
The cost is $\Oh(n^2)$ initial updates, then in total another $n^2$ updates and $n^2$ queries.
As the size of the LIS instance at any time is $\Oh(n^2)$ points, claimed conditional bound is proven. 
\end{proof}

\paragraph{Dynamic unweighted LIS supporting 1D-queries.}
For the unweighted case, we cannot use APSP conjecture anymore and need to switch to
the weaker OMv Conjecture~\ref{con:omv}.
Here we define $S_k$ as an embedding of matrix $A$ and vector $v_k$.

\begin{theorem}
There is no algorithm for dynamic (unweighted) LIS supporting 1D-queries with amortised update and query
time $\Oh(n^{1/3-\epsilon})$ for any constant $\epsilon>0$, unless Conjecture~\ref{con:omv} is false.
\end{theorem}
\begin{proof}
In this paragraph, we use $m$ to denote the size of the matrix and a vector.
Observe that the sum of weights in any embedding $S_k$ of Boolean matrix $A$ of size $m \times m$
and vector $v_k$ of size $m$ is $\Oh(m^3)$, as the largest weight is $\Oh(m)$.
Therefore, by replacing each weighted point with weight $w$ by a chain of $w$ unweighted points,
embedding $S_k$ becomes an instance of unweighted LIS on $\Oh(m^3)$ points.
Importantly, transforming $S_k$ into $S_k'$ still takes only $\Oh(m)$ updates,
since only $m$ points in $S_{b,k}$ needs to be deleted or inserted.
Observe that we cannot hope to beat the trivial algorithm for OMv running in cubic time
if we create an embedding of the whole input matrix, thus we will divide this matrix into smaller square pieces.

Suppose there is an algorithm for dynamic (unweighted) LIS supporting 1D-queries with amortised update
and query time $\Oh(n^{1/3-\epsilon})$, for some constant $\epsilon>0$,
when run on an instance with up to $n$ points.
Let $\epsilon'=1.5\epsilon$.
Divide $A$ into $m$ submatrices $A_{i,j}$ of size $m^{0.5} \times m^{0.5}$, and each vector $v_k$ into 
$m^{0.5}$ subvectors $v_{k,j}$ of size $m^{0.5}$.
Then $Av_k=u_k=(u_{k,1},\ldots,u_{k,m^{0.5}})$, where $|u_{k,i}|=m^{0.5}$ and
$u_{k,i}(j)=\max(0,\max_{l}((A_{i,l} \circ v_{k,l})(j))-1)$, as Boolean multiplication here translates to
the (max,$+$)-product.

Define $S_{k,i,j}$ as an embedding of $A_{i,j}$ and $v_{k,j}$, for $0 \leq j,i < m^{0.5}$ and $0 \leq k < n$.
Total cost of updates to create all $S_{1,\cdot,\cdot}$ is $\Oh(m^{3-\epsilon'})$, as there are
$m$ embeddings, each with $\Oh(m^{1.5})$ points.
To recover $Av_1$, we compute all $A_{i,j} \circ v_{1,j}$ by making $m^{1.5}$ queries,
each in an instance with $\Oh(m^{1.5})$ points, which in total takes time $\Oh(m^{2-\epsilon'})$.
Then by taking $m^{1.5}$ maximums we can compute $u_1=Av_1$.
Transforming all of the $S_{1,i,j}$ into $S_{2,i,j}$, for $0 \leq i,j < m^{0.5}$,
takes in total time $\Oh(m^{2-\epsilon'})$.
Then the process is repeated, allowing us to compute all $u_k$ one by one in an online fashion.
Total time is $\Oh(m^{3-\epsilon'})$, which is a contradiction under Conjecture~\ref{con:omv},
thus assumed algorithm for dynamic LIS could not exist.
\end{proof}

Finally, we note that in both presented lower bounds we can achieve simple trade-offs between bounds
on query and update time, as described in~\cite{AbboudD16}.

\bibliographystyle{plain}
\bibliography{biblio}

\end{document}